\documentclass{eptcs}
\providecommand{\event}{WRS 2009} 
\providecommand{\volume}{??}
\providecommand{\anno}{2009}
\providecommand{\firstpage}{1}
\providecommand{\eid}{??}

\usepackage{amsmath,amssymb,amsthm}

\newtheorem{example}{Example}
\newtheorem{definition}{Definition}

\newtheorem{lemma}{Lemma}
\newtheorem{theorem}{Theorem}
\newtheorem{corollary}{Corollary}

\def\titlerunning{Extending Context-Sensitivity in Term Rewriting}
\def\authorrunning{B.\ Gramlich \& F.\ Schernhammer}

\newcommand{\inff}{\mathsf{inf}}
\newcommand{\snd}{\mathsf{2nd}}
\newcommand{\s}{\mathsf{s}}
\newcommand{\app}{\mathsf{app}}
\newcommand{\take}{\mathsf{take}}
\newcommand{\nil}{\mathsf{nil}}
\newcommand{\topp}{\mathsf{top}}
\newcommand{\aaa}{\mathsf{a}}
\newcommand{\ttt}{\mathsf{t}}
\newcommand{\iii}{\mathsf{i}}

\begin{document}

\title{Extending Context-Sensitivity in Term Rewriting}
\author{Bernhard Gramlich and Felix Schernhammer%
\footnote{This author has been supported by the Austrian Academy of
  Sciences under grant 22.361.}
\institute{Institute of Computer Languages, Theory and Logic Group\\Vienna University of Technology}
\email{\{gramlich,felixs\}@logic.at}
}\maketitle

 \begin{abstract}
   We propose a generalized version of context-sensitivity in term
   rewriting based on the notion of ``forbidden patterns''.  The basic
   idea is that a rewrite step should be forbidden if the redex to be
   contracted has a certain shape and appears in a certain
   context. This shape and context is expressed through forbidden
   patterns. In particular we analyze the relationships among this
   novel approach and the commonly used notion of context-sensitivity
   in term rewriting, as well as the feasibility of rewriting with
   forbidden patterns from a computational point of view. The latter
   feasibility is characterized by demanding that restricting a
   rewrite relation yields an improved termination behaviour while
   still being powerful enough to compute meaningful
   results. Sufficient criteria for both kinds of properties in
   certain classes of rewrite systems with forbidden patterns are
   presented.
 \end{abstract}

%***************************************************************************
\section{Introduction and Overview}

Standard term rewriting systems (TRSs) are
well-known to enjoy nice logical and
closure properties. Yet, from an operational and computational point
of view, i.e., when using term rewriting as computational model, it is
also well-known that for non-terminating systems restricted versions
of rewriting obtained by imposing context-sensitivity and/or strategy
requirements may lead to better results (e.g., in terms of computing
normal forms, head-normal forms, etc.). 

One major goal when using reduction strategies and context restrictions
is to avoid non-terminating reductions. On the other hand 
the restrictions should not be too strong either, so that the
ability to compute useful results in the restricted rewrite systems
is not lost.
We introduce a novel approach to context restrictions relying
on the notion of ``forbidden patterns'', which generalizes
existing approaches and succeeds in handling examples in the
mentioned way (i.e., producing a terminating reduction relation
which is powerful enough to compute useful results) where others fail.

The following example motivates the use of reduction strategies and/or
context restrictions.

\begin{example}
\label{ex2nd}
Consider the following rewrite system, cf.\ e.g.\ \cite{ppdp01-lucas}:
\begin{eqnarray*}
  \mathsf{inf}(x)            & \rightarrow & x : \mathsf{inf}(\s(x)) \\
  \mathsf{\snd}(x : (y : zs)) & \rightarrow & y
\end{eqnarray*}
This TRS is non-terminating and not even weakly normalizing. Still
some terms like $\snd(\inff(x))$ are reducible to a normal form
while also admitting infinite reduction sequences. One goal of context
restrictions and reduction strategies is to 
restrict derivations in a way such that normal forms can be 
computed whenever they exist, while infinite reductions are
avoided.
\end{example} 

One way to address the problem of avoiding non-normalizing reductions in Example \ref{ex2nd} is 
the use of reduction strategies.
For 
instance for the class of (almost) orthogonal rewrite systems
(the TRS of Example \ref{ex2nd} is orthogonal),
always contracting all outermost redexes in parallel yields
a normalizing strategy (i.e.\ whenever a term can be reduced
to a normal form it is reduced to a normal form under this strategy)
\cite{odonnell}. Indeed, one can define a sequential reduction
strategy having the same property for an even wider class
of TRSs \cite{middeldorp}. One major drawback (or asset depending on
one's point of view) of using reduction
strategies, however, is that their use does not introduce 
new normal forms. This means that the set of normal forms
w.r.t.\ to some reduction relation is the same as the set
of normal forms w.r.t.\ to the reduction relation under some
strategy. Hence, strategies can in general not
be used to detect non-normalizing terms or to impose termination
on not weakly normalizing TRSs (with some exceptions
cf.\ e,g.\ \cite[Theorem 7.4]{middeldorp}). Moreover, 
the process of selecting a suitable redex w.r.t.\ to a reduction strategy
is often complex and 
may 
thus 
be
inefficient.

These shortcomings of reduction strategies led to the advent of proper 
restrictions of rewriting that
usually introduce new normal forms and select respectively forbid
certain reductions 
according to the syntactic structure of a redex and/or its
surrounding context.

The most well-known approach to context restrictions is context-sensitive
rewriting. There, a \emph{replacement map} $\mu$ specifies
the arguments $\mu(f) \subseteq \{1, \dots, ar(f)\}$ which can be
reduced for each function $f$. However, regarding Example \ref{ex2nd},
context-sensitive rewriting does not improve the situation, since
allowing the reduction of the second argument of
`$:$' leads to non-termination, while disallowing its reduction
leads to incompleteness in the sense that for instance a term like $\snd(\inff(x))$
cannot be normalized via the corresponding context-sensitive reduction relation, 
despite having a normal form in the unrestricted system.

Other ideas of context restrictions range from explicitly modeling
lazy evaluation
(cf.~e.g.~
\cite{%
toplas00-fokkink-et-al,%
wflp01-entcs02-lucas,%
wrs07-entcs08-schernhammer-gramlich}),
to imposing constraints on the order of argument evaluation of
functions 
(cf.\ e.g.\
\cite{%
popl85-futatsugi-et-al,%
wrla98-entcs}), 
and to combinations 
of these concepts, also with standard context-sensitive rewriting 
(cf.~e.g.~
\cite{%
ppdp01-lucas,%
lpar02-alpuente-et-al}). 
The latter generalized versions of context-sensitive rewriting are quite
expressive and powerful (indeed some of them can be used to restrict
the reduction relation of the TRS in Example \ref{ex2nd} in a way, so
that the restricted relation is terminating and still powerful
enough to compute (head-)normal forms),
but on the other hand tend to be hard to
analyze and understand, due 
the subtlety of the strategic information specified. 

The approach we present in this paper is simpler in that its
definition only relies on matching and simple comparison of
positions rather than on laziness or prioritizing 
the evaluation of certain arguments
of functions over others.
In order to reach the goal of restricting the reduction relation in
such a way that it is terminating
while still being powerful enough to compute useful results,
we provide a method to verify
termination of a reduction relation restricted by our approach (Section \ref{sec:proving-termination}) as
well as a criterion which guarantees that normal forms computed
by the restricted system are head-normal forms of the
unrestricted system (Section \ref{sec:computing-meaningful-results}).

Recently
it turned out that,
apart from using context-sensitivity as computation model for
standard term rewriting (cf.~e.g.~ \cite{ic02-lucas,jflp98-lucas}), 
context-sensitive rewrite systems naturally also appear 
as intermediate representations in many areas relying on
transformations, such as program transformation and termination
analysis of rewrite systems 
with conditions
\cite{hosc08-duran-et-al,techrep09-schernhammer-gramlich} / under
strategies 
\cite{rta09-endrullis-hendriks}.

This suggests that apart from using restrictions as guidance
and thus as operational model for rewrite derivations, a general, 
flexible and well-understood framework of restricted 
term rewriting going beyond context-sensitive rewriting
may be useful as a valuable tool in many other areas, too.

The major problem in building such a framework is that imposing
context restrictions on term rewriting in general invalidates
the closure properties of term rewriting relations, i.e., stability under
contexts and substitutions. Note that in the case of context-sensitive
rewriting \`a la 
\cite{jflp98-lucas,ic02-lucas} 
only stability under contexts is lost.

In this work we will sketch and discuss a generalized approach 
to context-sensitivity (in the sense of
\cite{jflp98-lucas,ic02-lucas}) relying on \emph{forbidden patterns} 
rather 
than on forbidden arguments of functions. From a systematic point
of view we see the following design decisions to be made. 

\begin{itemize}
\item What part of the context of a (sub)term is relevant to decide
whether the (sub)term may be reduced or not?
\item In order to specify the restricted reduction relation, is it
  better/advantageous to explicitly define the allowed or the
  forbidden part of the context-free reduction relation?
\item What are the forbidden/allowed entities, for instance whole
  subterms, contexts, positions, etc.?
\item Does it depend on the shape of the considered subterm itself (in
  addition to its outside context) whether it should forbidden or not
  (if so, stability under substitutions may be lost)?
\item Which restrictions on forbidden patterns seem appropriate
    (also 
  w.r.t.\ practical feasibility) in order to guarantee certain desired
  closure and preservation properties.
\end{itemize}
The remainder of the paper is structured as follows. In Section
\ref{sec:preliminaries} we briefly recall some basic notions and
notations. Rewriting with forbidden patterns is defined, discussed and
exemplified in Section \ref{sec:rewriting-with-forbidden patterns}. In
the main Sections \ref{sec:computing-meaningful-results} and
\ref{sec:proving-termination} we develop some theory about the
expressive power of rewriting with forbidden patterns (regarding the ability
to compute original (head-)normal forms), and about how to prove
ground termination for such systems via a constructive transformational
approach. Crucial aspects are illustrated with the two running
Examples \ref{ex2nd} and \ref{ex_app}. Finally, in Section
\ref{sec:conclusion-and-related-work} we summarize our approach and
its application in the examples, discuss its relationship to previous
approaches and briefly touch the important perspective and open
problem of (at least partially) automating the generation of suitable
forbidden patterns in practice.%
\footnote{Due to lack of space the obtained results are presented
  without proofs. The latter can be found in the full technical report
version of the paper, cf.\
\texttt{http://www.logic.at/staff/\{gramlich,schernhammer\}/}.} 

%***************************************************************************
\section{Preliminaries}
\label{sec:preliminaries}

We assume familiarity with the basic notions and notations in term
rewriting, cf.~e.g.~\cite{book98-baader-nipkow}, \cite{BeKlVr03}. 

Since we develop our approach in a many-sorted setting, we recall a
few basics on many-sorted equational reasoning (cf.\ e.g.\
\cite{BeKlVr03}). 
A many-sorted signature $\mathcal{F}$ is a pair $(S, \Omega)$ where
$S$ is a set of sorts and $\Omega$ is a family of (mutually disjoint)
sets of typed function symbols:
$\Omega = (\Omega_{\omega, s} \mid \omega \in S^*, s \in S)$.
We also say, $f$ is of type $\omega \rightarrow s$ (or just $s$ if
$\omega = \emptyset$) if $f \in \Omega_{\omega, s}$.
$V = (V_s \mid s \in S)$ is a family of (mutually disjoint) countably
infinite sets 
of typed variables (with $V \cap \Omega = \emptyset$).
The set $\mathcal{T}(\mathcal{F}, V)_s$ of (well-formed) terms of sort $s$ is
the least set containing $V_s$, and whenever $f \in \Omega_{(s_1,
  \dots,  s_n), s}$ 
and $t_i \in \mathcal{T}(\mathcal{F}, V)_{s_i}$ for all $1 \leq i \leq n$, then
$f(t_1, \dots, t_n) \in \mathcal{T}(\mathcal{F}, V)_s$.
The sort of a term $t$ is denoted by $sort(t)$.
Rewrite rules are pairs of terms $l \rightarrow r$ where $sort(l) = sort(r)$.
Subsequently, we make the types of terms and rewrite rules explicit
only if they are relevant. Throughout the paper $x, y, z$
represent (sorted) variables.

Positions are possibly empty sequences of natural numbers (the empty sequence is
denoted by $\epsilon$). We use the standard 
partial order $\leq$ on positions given by $p \leq q$ if
there is some position $p'$, such that $p.p' = q$ (i.e., $p$
is a prefix of $q$).
 $Pos(s)$ ($Pos_\mathcal{F}(s)$) denotes the set of (non-variable) positions
of a term $s$. By $s \overset{p}{\rightarrow} t$ we mean rewriting at
position $p$. Given a TRS $\mathcal{R} = (\mathcal{F}, R)$ we 
partition $\mathcal{F}$ into the set $D$ of defined function symbols,
which are those that occur as root symbols of left-hand sides of
rules in $R$, and the set $C$ of constructors (given by $\mathcal{F}
\setminus D$). For TRSs $\mathcal{R} = (\mathcal{F}, R)$ we sometimes
confuse $\mathcal{R}$ and $R$, e.g., by omitting the signature.

%***************************************************************************
\section{Rewriting with Forbidden Patterns}
\label{sec:rewriting-with-forbidden patterns}

In this section we define a generalized approach to rewriting with 
context restrictions relying on term patterns to specify forbidden
subterms/super\-terms/positions rather than on a replacement map as
in context-sensitive rewriting. 

\begin{definition}[forbidden pattern]
A \emph{forbidden pattern} (w.r.t.~to a signature $\mathcal{F}$) 
is a triple $\langle t, p, \lambda \rangle$, where $t \in \mathcal{T}(\mathcal{F}, V)$ is a
term, $p$ a position from $Pos(t)$ and $\lambda \in \{h, b, a\}$.
\end{definition}

The intended meaning of the last component $\lambda$ is to indicate
whether the pattern forbids reductions
\begin{itemize}
\item exactly at position $p$, but not outside (i.e., strictly above
  or parallel to $p$) or strictly below 
  -- ($h$ for here), or 
\item strictly below $p$, but not at or
  outside $p$ -- ($b$ for below), or
\item strictly above position $p$, but not at, below or parallel to $p$
  --  ($a$ for above). 
\end{itemize}

Abusing notation we sometimes say a forbidden pattern is linear, 
unifies with some term etc.\ when we actually mean that the term in the
first component of a forbidden pattern has this property.

We denote a \emph{finite} set of forbidden patterns for a signature
$\mathcal{F}$ by $\Pi_{\mathcal{F}}$ or just $\Pi$ if $\mathcal{F}$ is clear
from the context or irrelevant. For brevity, patterns of the
shape $\langle \_, \_, h/b/a \rangle$ are also called
$h/b/a$-patterns, or $here/below/above$-patterns.%
\footnote{
Here and subsequently we use a wildcard notation for forbidden
patterns. For instance, $\langle \_, \_, i\rangle$ stands for 
$\langle t, p, i\rangle$ where $t$ is some term 
and $p$ some position in $t$ of no further relevance.}

Note that if for a given term $t$ we want to specify more than just
one restriction by a forbidden pattern, this can easily be achieved by
having several triples of the shape $\langle t,\_,\_ \rangle$. 

In contrast to context-sensitive rewriting, where a replacement map
defines the allowed part of the reduction, the patterns are supposed
to explicitly define its \emph{forbidden} parts, thus implicitly
yielding allowed reduction steps as those that are not forbidden.
 
\begin{definition}[forbidden pattern reduction relation]
\label{fp_relation}
Let $\mathcal{R} = (\mathcal{F}, R)$ be a TRS with forbidden patterns
$\Pi_{\mathcal{F}}$. The \emph{forbidden pattern reduction relation}
$\rightarrow_{\mathcal{R}, \Pi_{\mathcal{F}}}$, or
$\rightarrow_{\Pi}$ for short, induced by some set of forbidden patterns $\Pi$ and $\mathcal{R}$,
is given by 
$s \rightarrow_{\mathcal{R}, \Pi_{\mathcal{F}}} t$
if $s \overset{p}{\rightarrow}_{\mathcal{R}} t$ for some $p \in
Pos_{\mathcal{F}}(s)$ such that there is no pattern $\langle u, q, \lambda \rangle \in
\Pi_{\mathcal{F}}$, no context $C$ and no position $q'$ with 
\begin{itemize}
\item $s = C[u\sigma]_{q'}$ and $p = q'.q$, if $\lambda = h$, 
\item $s = C[u\sigma]_{q'}$ and $p > q'.q$, if $\lambda = b$, and
\item $s = C[u\sigma]_{q'}$ and $p < q'.q$, if $\lambda = a$.
\end{itemize}
\end{definition}

Note that for a finite rewrite system
$\mathcal{R}$ (with finite signature $\mathcal{F}$) and a finite set of 
forbidden patterns $\Pi_{\mathcal{F}}$ it is decidable whether $s
\rightarrow_{\mathcal{R}, \Pi_{\mathcal{F}}} t$ for terms $s$ and $t$. 
We write $(\mathcal{R}, \Pi)$ for rewrite systems with 
associated
forbidden patterns. Such a rewrite system $(\mathcal{R}, \Pi)$
is said to be $\Pi$-terminating (or just terminating if no confusion
arises) if $\rightarrow_{\mathcal{R}, \Pi}$ is well-founded.
We also speak of $\Pi$-normal forms instead of
$\rightarrow_{\mathcal{R}, \Pi}$-normal forms.

Special degenerate cases of $(\mathcal{R}, \Pi)$ include e.g.\
$\Pi = \emptyset$ where $\rightarrow_{\mathcal{R}, \Pi} =
\rightarrow_{\mathcal{R}}$, and $\Pi = \{ \langle l,\epsilon,h \rangle \;|\; l
\rightarrow r  \in R\}$ where $\rightarrow_{\mathcal{R}, \Pi} =
\emptyset$. 

In the sequel we use the notions of \emph{allowed} and \emph{forbidden}
(by $\Pi$) redexes. A redex $s|_p$ of a term $s$ is allowed
if $s \overset{p}{\rightarrow_{\Pi}} t$ for some term $t$,
and forbidden otherwise.

\begin{example}
\label{ex2nd2}
Consider the TRS from Example \ref{ex2nd}. If $\Pi = \{(x : (y :
\inff(z)), 2.2, h)\}$, 
then $\rightarrow_{\Pi}$ can automatically be shown to be terminating. Moreover,
$\rightarrow_{\Pi}$ is powerful enough to compute original head-normal forms if
they exist (cf. Examples \ref{ex2nd3} and \ref{ex2nd4} below).
\end{example}

\begin{example}
\label{ex_app}
Consider the non-terminating TRS $\mathcal{R}$ given by
\begin{equation}
 \nonumber\begin{tabular}[b]{r@{ $\rightarrow$ }l@{\;\;\;\;\;\;}r@{ $\rightarrow$ }l}
    $\take(0, y:ys)$ & $y$           & $\app(\nil, ys)$ & $ys$ \\
$\take(\s(x), y:ys)$ & $\take(x, ys)$ & $\app(x:xs, ys)$ & $x : \app(xs, ys)$ \\
	 $\take(x, \nil)$ & $0$       & $\inff(x)$ & $\inff(\s(x))$
\end{tabular}
\end{equation}
with two sorts $S = \{Nat, NatList\}$, where 
the types of function symbols are as follows:
$\nil \colon NatList$,
$0: Nat$,
$s: Nat \rightarrow Nat$,
$:$ is of type $Nat, NatList \rightarrow NatList$, 
$\inff: Nat \rightarrow NatList$, 
$\app: NatList, NatList \rightarrow NatList$
and 
$take: Nat, NatList \rightarrow Nat$.
If one restricts rewriting in $\mathcal{R}$ via
 $\Pi$ given by
\begin{equation}
 \nonumber\begin{tabular}[b]{r@{\hspace*{3ex}}c@{\hspace*{3ex}}l}
    $\langle x : \inff(y), 2, h\rangle\;\;$ & 
    $\;\;\langle x : \app(\inff(y), zs), 2.1 , h\rangle\;\;$ & 
    $\;\;\langle x : \app(y : \app(z, zs), us), 2, h\rangle$,
\end{tabular}
\end{equation}
then $\rightarrow_{\Pi}$ is terminating and still
every well-formed ground term can be normalized with
the restricted relation $\rightarrow_{\Pi}$ (provided the term is normalizing). 
See Examples \ref{ex_app_compl} and \ref{ex_app_term} below for justifications of these claims.
\end{example}

Several well-known approaches to restricted term rewriting as well
as to rewriting guided by reduction strategies occur as special cases
of rewriting with forbidden patterns. In the following we provide some examples.
Context-sensitive rewriting, where a replacement map $\mu$ specifies
the arguments $\mu(f) \subseteq \{1, \dots, ar(f)\}$ which can be
reduced for each function $f$,
arises as special case of
rewriting with forbidden patterns by defining $\Pi$ to
contain for each function symbol $f$ and each 
$j \in \{1, \dots, ar(f)\} 
\setminus \mu(f)$ the forbidden patterns $(f(x_1, \ldots, x_{ar(f)}), j, h)$ 
and $(f(x_1, \ldots, x_{ar(f)}), j, b)$. 

Moreover, with forbidden patterns it is also possible to simulate
position-based 
reduction strategies such as innermost and outermost rewriting. The innermost
reduction relation of a TRS $\mathcal{R}$ coincides with the forbidden pattern
reduction relation if one uses the forbidden patterns $ \langle l,
\epsilon, a \rangle$ for the left-hand sides $l$ of each rule of $\mathcal{R}$. Dually, if patterns
$(l, \epsilon, b)$ are used, the forbidden pattern reduction relation coincides
with the \emph{outermost} reduction relation w.r.t.\ $\mathcal{R}$.

However, note that more complex layered combinations of the
  aforementioned approaches,  
such as innermost context-sensitive rewriting cannot be modeled by forbidden
patterns as proposed in this paper.

Still, the definition of forbidden patterns and rewriting with forbidden patterns
is 
rather
general and leaves many parameters open. In order to make
this approach feasible in practice, it is necessary to identify 
interesting classes of forbidden patterns that yield a reasonable
trade-off between power and simplicity.
For these interesting classes of forbidden patterns we need methods 
which guarantee that the results (e.g.\ normal forms) computed by rewriting with forbidden 
patterns are meaningful, in the sense that they
have some natural correlation with the actual results obtained by
unrestricted rewriting. For instance, it is desirable that normal forms
w.r.t.\ the restricted rewrite system are original head-normal
forms. In this case one can use the restricted reduction relation to
compute original normal forms (by an iterated process) whenever they
exist (provided that the TRS in question is left-linear, confluent and the
restricted reduction relation is terminating) (cf.\ Section \ref{sec:computing-meaningful-results} below for details).
We define a criterion ensuring that normal forms w.r.t.\ the restricted system
are original head-normal forms in the following section.

%***************************************************************************
\section{Computing Meaningful Results}
\label{sec:computing-meaningful-results}

We are going to use canonical context-sensitive
rewriting as defined in \cite{jflp98-lucas,ic02-lucas} as an inspiration for our approach.
There, for a given (left-linear) rewriting system $\mathcal{R}$ certain 
restrictions on the associated replacement map $\mu$ guarantee that
$\rightarrow_{\mu}$-normal forms are $\rightarrow_{\mathcal{R}}$-head-normal-forms.
Hence, results computed by $\rightarrow_{\mu}$ and $\rightarrow_{\mathcal{R}}$
share the same root symbol.

The basic idea is that reductions that are essential to create a more outer redex
should not be forbidden. In the case of context-sensitive rewriting 
this is guaranteed by demanding that whenever an $f$-rooted term $t$
occurs (as subterm) 
in the left-hand side of a rewrite rule and has a non-variable direct
subterm $t|_i$, then
$i \in \mu(f)$.

It turns out that for rewriting with forbidden patterns 
severe restrictions on the shape of the
patterns are necessary in order to obtain results similar to
the ones for canonical context-sensitive rewriting in \cite{jflp98-lucas}.
First, no forbidden patterns of the shape $\langle \_, \epsilon, h \rangle$ or
$\langle \_,\_,a \rangle$ may be used as they are in general not compatible with the desired
root-normalizing behaviour of our forbidden pattern rewrite system.

Moreover, for each pattern $\langle t, p, \_ \rangle$ we demand that 
\begin{itemize}
\item $t$ is linear,
\item $p$ is a variable or maximal (w.r.t.\ to the prefix ordering $\leq$ on positions)  
non-variable position in $t$, and
\item for each position $q \in Pos(t)$ with $q || p$ we have $t|_q \in
  V$. 
\end{itemize}

We call the class of patterns obtained by the above restrictions \emph{simple patterns}.

\begin{definition}[simple patterns]
\label{def_simple_patterns}
A set $\Pi$ of forbidden patterns is called \emph{simple} if it does not contain
patterns of the shape $\langle \_, \epsilon, h \rangle$ or $\langle \_,\_,a \rangle$ and for every pattern
$(t, p, \_) \in \Pi$ it holds that $t$ is linear, $t|_p \in V$ or $t|_p = f(x_1, \dots, x_{ar(f)})$
for some function symbol $f$, and for each position $q \in Pos(t)$
with $q || p$ we have that $t|_q$ is a variable. 
\end{definition}

Basically these syntactical properties of forbidden patterns
are necessary to ensure that reductions which are essential to enable
other, more outer reductions are not forbidden. Moreover, these properties, 
contrasting those defined in Definition \ref{def_canonicity} below, are
independent of any concrete rewrite system. 

The forbidden patterns of the TRS ($\mathcal{R}, \Pi$) in Example \ref{faa}
below
are not simple, since the patterns contain terms with parallel non-variable
positions. This is the reason why it is not possible 
to head-normalize terms (w.r.t\ $\mathcal{R}$) with $\rightarrow_{\Pi}$: 

\begin{example}
\label{faa}
Consider the TRS $\mathcal{R}$ given by 
\begin{equation}
 \nonumber\begin{tabular}[b]{r@{ $\rightarrow$ }l@{\;\;\;\;\;\;}r@{ $\rightarrow$ }l}
    $f(b, b)$ & $g(f(a, a))$ & $a$ & $b$
\end{tabular}
\end{equation}
and forbidden patterns $\langle f(a, a), 1, h \rangle$ and
$\langle f(a, a), 2, h \rangle$. $f(a, a)$ is linear and
$1$ and $2$ are maximal positions (w.r.t.\ $\leq$) within this term.
However, positions $1$ and $2$ are both non-variable
and thus e.g.\ for $\langle f(a, a), 1, h \rangle$
there exists a position $2 || 1$ such that $f(a, a)|_2 = a \not\in
V$. Hence, $\Pi$ is too restrictive 
to compute all $\mathcal{R}$-head-normal forms in this
example. Indeed,
$f(a,a) \rightarrow_{\mathcal{R}}^* f(b,b) \rightarrow_{\mathcal{R}}
g(f(a,a))$ where the latter term is a $\mathcal{R}$-head-normal form.

The term
$f(a, a)$ is a $\Pi$-normal form, although it is not a 
head-normal form (w.r.t.\ $\mathcal{R}$). Note also that
the (first components of) forbidden patterns are not unifiable with the left-hand
side of the rule that is responsible for the (later) possible root-step
when reducing $f(a, a)$, not even if the forbidden subterms in the patterns are
replaced by fresh variables. 
\end{example}

Now we are ready to define canonical rewriting with
forbidden patterns within the class of simple forbidden patterns.
To this end, we demand that patterns do not overlap
with left-hand sides of rewrite rules in a way such that
reductions necessary to create a redex might be forbidden.

\begin{definition}[canonical forbidden patterns]
\label{def_canonicity}
Let $\mathcal{R} = (\mathcal{F}, R)$ 
be a TRS with 
simple forbidden patterns $\Pi_{\mathcal{F}}$ (w.l.o.g.\ we assume that $R$ and
$\Pi_{\mathcal{F}}$ have no variables in common). 
Then, $\Pi_{\mathcal{F}}$ 
is \emph{$\mathcal{R}$-canonical} (or just \emph{canonical}) if
the following holds for all rules $l \rightarrow r
\in R:$
\begin{enumerate}
\item \label{rule overlaps pattern} There is no pattern $(t, p, \lambda)$ such that 
\begin{itemize}
\item  $t'|_q$ and $l$ unify for some $q \in Pos_\mathcal{F}(t)$ where $t' = t[x]_p$ and $q >
  \epsilon$, and
\item there exists
  a position $q' \in Pos_{\mathcal{F}}(l)$ with $q.q' = p$ for $\lambda = h$
  respectively $q.q' > p$ for $\lambda = b$.
\end{itemize}

\label{pattern overlaps rule} \item There is no pattern $(t, p, \lambda)$ such that 
\begin{itemize}
\item $t'$ and $l|_q$ unify for some $q \in Pos_\mathcal{F}(l)$ where $t' = t[x]_p$, and
\item there exists a position $q'$ with $q.q' \in Pos_{\mathcal{F}}(l)$ 
    and $q' = p$ for $\lambda = h$ respectively $q' > p$ for $\lambda = b$.
\end{itemize}
Here, $x$ denotes a fresh variable.
\end{enumerate}
\end{definition}

\begin{example}
Consider the TRS $\mathcal{R}$ given by the single rule
\begin{eqnarray*}
l = f(g(h(x))) & \rightarrow &  x = r\,.
\end{eqnarray*}
Then, $\Pi = \{\langle t,p,h \rangle\}$ with $t = g(f(a))$, $p =
1.1$ is not 
canonical since $t[x]_p|_q = g(f(y))|_1 = f(y)$ and $l$ unify
where $q = q' = 1$ and thus $q.q' = p$ (hence $root(l|_{q'}) = g$).
Moreover, also $\Pi = \{\langle t, p, h\rangle\}$
with 
$t = g(i(x))$, $p = 1$
is
not canonical, since $l|_q = g(h(x))$ and $t[x]_p = f(y)$ unify for $q = 1$ and
$q.p = 1.1$ is a non-variable position in $l$.   

On the other hand, $\Pi = \{\langle g(g(x)), 1.1, h \rangle\}$ is canonical. Note that
all of the above patterns are simple.
\end{example}

In order to prove that normal forms obtained by rewriting with simple and canonical forbidden
patterns are actually head-normal forms w.r.t.\ unrestricted rewriting, and
also to provide more intuition on canonical rewriting with forbidden patterns, we define
the notion of a \emph{partial redex} (w.r.t.\ to a rewrite system $\mathcal{R}$) as a term 
that is matched by a non-variable term $l'$ which in turn matches the left-hand side of
some rule of $\mathcal{R}$. We call $l'$ a \emph{witness} for the partial match.

\begin{definition}[Partial redex]
Given a rewrite system $\mathcal{R} = (\mathcal{F}, R)$, a \emph{partial redex}
is a term $s$ that is matched by a non-variable term $l'$ which in turn matches
the left-hand side of some rule in $R$. The (non-unique) term $l'$ is called \emph{witness}
for a partial redex $s$.
\end{definition}

Thus, a partial redex can be viewed as a candidate for a future reduction step, 
which can only be performed if the redex has actually been created through more
inner reduction steps. 
Hence, the idea of canonical rewriting with forbidden patterns could
be reformulated as guaranteeing that the reduction of subterms of partial redexes is allowed
whenever these reductions are necessary to create an actual redex.

\begin{lemma}
\label{lem_canonicity}
Let $\mathcal{R} = (\mathcal{F}, R)$ be a \emph{left-linear} TRS with
\emph{canonical} (hence, in particular \emph{simple})
forbidden patterns $\Pi_{\mathcal{F}}$. Moreover, let $s$ be a partial redex w.r.t.\ to the left-hand side
of some rule $l$ with witness $l'$ such that $l|_p \not\in V$ but $l'|_p \in V$.
Then in the term $C[s]_q$ the position $q.p$ is allowed by $\Pi_{\mathcal{F}}$ for reduction provided that
$q$ is allowed for reduction.
\end{lemma}

\begin{theorem}
\label{prop_complete}
Let $\mathcal{R} = (\mathcal{F}, R)$ be a \emph{left-linear} TRS with
\emph{canonical} (hence in particular \emph{simple})
forbidden patterns $\Pi_{\mathcal{F}}$. Then $\rightarrow_{\mathcal{R}, \Pi_{\mathcal{F}}}$-normal forms
are $\rightarrow_{\mathcal{R}}$-head-normal forms.
\end{theorem}

Given a left-linear and confluent rewrite system $\mathcal{R}$ and a set 
of canonical forbidden patterns $\Pi$ such that $\rightarrow_{\Pi}$
is well-founded, one can thus normalize a term $s$ (provided that $s$ is
normalizing) by computing
the $\rightarrow_{\Pi}$-normal form $t$ of $s$ which is $\mathcal{R}$-root-stable
according to Theorem \ref{prop_complete},  and then do the same
recursively for the immediate subterms of $t$. Confluence of $\mathcal{R}$
assures that the unique normal form of $s$ will indeed be computed this way. 

\begin{example}
\label{ex2nd3}
As the forbidden pattern defined in Example \ref{ex2nd2} is (simple and)
canonical, 
Theorem \ref{prop_complete} yields that $\rightarrow_{\mathcal{R}, \delta}$-normal 
forms are $\rightarrow_{\mathcal{R}}$-head-normal forms. 
For instance we get $2nd(\inff(0)) \rightarrow_{\Pi}^* \s(0)$.
\end{example}

\begin{example}
\label{ex_app_compl}
Consider the TRS with $\mathcal{R}$ and forbidden patterns $\Pi$ from
Example \ref{ex_app}. We will prove below that $\mathcal{R}$ is
$\Pi$-terminating (cf.\ Example \ref{ex_app_term}).  

Furthermore we are able to 
show that
every well-formed ground term that is reducible to a normal form
in $\mathcal{R}$ is reducible to the same normal form with
$\rightarrow_{\mathcal{R}, \Pi}$ and that every
$\rightarrow_{\mathcal{R}}$-normal form is root-stable
w.r.t.\ $\rightarrow_{\mathcal{R}}$.
\end{example}

%***************************************************************************
\section{Proving Termination}
\label{sec:proving-termination}

We provide another example of a result on a restricted class of forbidden patterns, 
this time concerning termination. We exploit the fact that, given
a finite signature and linear $h$-patterns, 
a set of allowed contexts complementing each forbidden one
can be constructed. Thus, we can transform a rewrite system with this kind
of forbidden patterns into a standard (i.e., context-free) one by 
explicitly instantiating and embedding all rewrite rules (in a minimal way) 
in contexts (including a designated
$\topp$-symbol representing the \emph{empty} context) such that 
rewrite steps in these contexts are allowed.

To this end we propose a transformation that proceeds by iteratively instantiating
and embedding rules in a minimal way. This is to say that the used substitutions
map variables only to terms of the form $f(x_1, \dots, x_{ar(f)})$ and the
contexts used for the embeddings have the form $g(x_1, \dots, x_{i-1}, \Box, x_{i+1}, x_{ar(f)})$
for some function symbols $f \in \mathcal{F}$, $g \in \mathcal{F} \uplus \{\topp\}$ 
and some argument position $i$ of $f$ (resp.\ $g$). It is important
to keep track of the position of the initial rule inside the
embeddings. 
Thus we associate to each rule introduced by the transformation a
position pointing to the embedded original rule. To all initial rules
of $\mathcal{R}$ we thus associate $\epsilon$. 

Note that it is essential
to consider a new unary function symbol $\topp_s$ for every sort $s \in S$
(of type $s \rightarrow s$) representing the empty context. This is
illustrated by the following example. 

\begin{example}
Consider the TRS given by
\begin{equation}
 \nonumber\begin{tabular}[b]{r@{ $\rightarrow$ }l@{\;\;\;\;\;\;}r@{ $\rightarrow$ }l}
    $a$ & $f(a)$ & $f(x)$ & $x$
\end{tabular}
\end{equation}
with $\mathcal{F} = \{a, f\}$ and the set of forbidden patterns
$\Pi = \{\langle f(x), 1, h\}\rangle\}$.
This system is not $\Pi$-terminating as we have
\begin{equation*}
a \rightarrow_{\Pi} f(a) \rightarrow_{\Pi} a \rightarrow_{\Pi} \dots
\end{equation*}
Whether a subterm $s|_p = a$ is allowed for reduction by $\Pi$ depends on
its context. Thus, according to the idea of our transformation we try
to identify all contexts $C[a]_p$ such that the reduction of $a$ at position
$p$ is allowed by $\Pi$. However, there is no such (non-empty) context, although
$a$ may be reduced if $C$ is the empty context. Moreover, there cannot be
a rule $l \rightarrow r$ in the transformed system where $l = a$, since
that would allow the reduction of terms that might be forbidden by $\Pi$.
Our solution to this problem is to introduce a new function symbol $\topp$
explicitly representing the empty context. Thus, in the example the
transformed system will contain a rule $\topp(a) \rightarrow \topp(f(a))$.
\end{example}

Abusing notation we subsequently use only one $\topp$-symbol, while we actually
mean the $\topp_s$-symbol of the appropriate sort. Moreover, in the following 
by rewrite rules we always mean rewrite rules with an associated
(embedding) position,
unless stated otherwise. All forbidden patterns used in this section
(particularly in the lemmata)
are linear here-patterns. We will make this general assumption explicit
only in the more important results. 

\begin{definition}[instantiation and embedding]
\label{def_t}
Let $\mathcal{F} = (S, \Omega)$ be a signature, let $\langle l \rightarrow r, p\rangle$ be a rewrite rule
of sort $s$ over $\mathcal{F}$ and let $\Pi$ be a 
set of forbidden patterns (linear, $h$). The set of minimal instantiated
and embedded rewrite rules $T_{\Pi}(\langle l \rightarrow r, p\rangle)$ 
(or just $T(\langle l \rightarrow r, p\rangle)$)
is
$T^i_{\Pi}(\langle l \rightarrow r, p\rangle) \uplus T_{\Pi}^e(\langle l \rightarrow r, p\rangle)$
where
\begin{eqnarray*}
T^e(\langle l \rightarrow r, p\rangle) & = & \{ \langle C[l] \rightarrow C[r], i.p \rangle \mid
C = f(x_1, \dots, x_{i-1}, \Box, x_{i+1}, \dots, x_{ar(f)}),\\
& & f \in \Omega_{(s_1, \dots, s_{i-1}, s, s_{i+1}, \dots, s_{ar(f)}), s'},
f \in \mathcal{F} \uplus \{\topp_s \mid s \in S\}, 
i \in \{1, \dots, ar(f)\},\\
& & \exists \langle u, o, h \rangle \in \Pi. u|_q \theta = l \theta \wedge q \not= \epsilon \wedge o = q.p \}\\
T^i_{\Pi}(\langle l \rightarrow r, p\rangle) & = & \{ \langle l \sigma \rightarrow r \sigma, p\rangle \mid
x \sigma = f(x_1, \dots, x_{ar(f)}), sort(x) = sort(f(x_1, \dots x_{ar(f)})), \\
& &f \in \mathcal{F}, y \not= x \Rightarrow y \sigma = y, x \in RV_{\Pi}(l, p)\} 
\end{eqnarray*}
and $RV_{\Pi}(l, p) = \{x \in Var(l) \mid \exists \langle u, o, h \rangle \in \Pi. \theta = mgu(u, l|_q) \wedge q.o = p
\wedge x \theta \not \in V \}$.

We also call the elements of $T(\langle l \rightarrow r, p\rangle)$ the one-step $T$-successors
of $\langle l \rightarrow r, p\rangle$. The reflexive-transitive closure of the one-step
$T$-successor relation is the many-step $T$-successor relation or just $T$-successor
relation. We denote the set of all many-step $T$-successors of a rule $\langle l \rightarrow r, p\rangle$
by $T^*(\langle l \rightarrow r, p\rangle)$.
\end{definition}

The set $RV_{\Pi}(l, p)$ of ``relevant variables'' is relevant in the
sense that their instantiation might 
contribute to a matching by some (part of a) forbidden pattern term.

Note that in the generated rules $\langle l' \rightarrow r',
p'\rangle$ in $T_{\Pi}(\langle l \rightarrow r, p\rangle)$, a fresh
$\topp_s$-symbol can only occur at the root of both $l'$ and $r'$ or
not at all, according to the construction in Definition \ref{def_t}.

\begin{example}
  \label{ex-for-instantiaion-and-embedding}
Consider the TRS $(\mathcal{R}, \Pi)$ where $\mathcal{R} = (\{a,f,g\},\{f(x) \rightarrow g(x)\})$ 
and the forbidden patterns $\Pi$ are given by $\{\langle g(g(f(a))), 1.1, h\rangle \}$. 
$T(\langle f(x) \rightarrow g(x), \epsilon\rangle)$
consists of the following rewrite rules.
\begin{eqnarray}
\label{1} \langle f(f(x)) & \rightarrow & g(f(x)), \epsilon\rangle \\
\label{2} \langle f(g(x)) & \rightarrow & g(g(x)), \epsilon\rangle \\
\label{3} \langle f(a) & \rightarrow & g(a), \epsilon\rangle \\
\label{4} \langle f(f(x)) & \rightarrow & f(g(x)), 1\rangle \\
\label{5} \langle g(f(x)) & \rightarrow & g(g(x)), 1\rangle
\end{eqnarray}

Note that $RV_{\Pi}(f(x), \epsilon) = \{x\}$ because $g(g(f(a)))_{1.1} = f(a)$
unifies with $f(x)$ and mgu $\theta$ where $x \theta = a \not \in V$.
On the other hand $RV_{\Pi}(f(f(x)), 1) = \emptyset$. 

\end{example}

\begin{lemma}[finiteness of instantiation and embedding]
\label{lem_term}
Let $\langle l \rightarrow r, p\rangle$ be a rewrite rule and let $\Pi$
be a set of forbidden patterns. The set of (many-step) instantiations
and embeddings of $\langle l \rightarrow r, p\rangle$ (i.e.\ $T^*(\langle l \rightarrow r, p\rangle))$
 is finite.
\end{lemma}

The transformation we are proposing proceeds by iteratedly instantiating
and embedding rewrite rules. The following definitions identify
the rules for which no further instantiation and embedding is needed.

\begin{definition}[$\Pi$-stable]
\label{def_stable}
Let $\langle l \rightarrow r, p\rangle$ be a rewrite rule and
let $\Pi$ be a set of forbidden patterns. 
$\langle l \rightarrow r, p\rangle$
is $\Pi$-stable ($stb_{\Pi}(\langle l \rightarrow r, p\rangle)$ for short) 
if there is no context
$C$ and no substitution $\sigma$ such that $C[l \sigma]_q|_{q'} = u \theta$ and $q.p = q'.o$
for any forbidden pattern $\langle u, o, h \rangle \in \Pi$ and any $\theta$.
\end{definition}

Note that $\Pi$-stability is effectively decidable (for finite signatures and finite $\Pi$), 
since only contexts
and substitutions involving terms not exceeding a certain depth
depending on $\Pi$ need to be considered.

\begin{definition}[$\Pi$-obsolete]
Let $\langle l \rightarrow r, p\rangle$ be a rewrite rule and
let $\Pi$ be a set of forbidden patterns. 
$\langle l \rightarrow r, p\rangle$
is $\Pi$-obsolete ($obs_{\Pi}(\langle l \rightarrow r, p\rangle)$ for short)
if there is a forbidden pattern $\Pi = \langle u, o, h \rangle$ such
that $l|_q = u \theta$ and $p = q.o$.
\end{definition}

In Example \ref{ex-for-instantiaion-and-embedding}, the rules (\ref{1}),
(\ref{2}) and (\ref{4}) are $\Pi$-stable, while 
rules (\ref{3}) and (\ref{5}) would be processed further. After two more steps
e.g.\ a rule $\langle g(g(f(a))) \rightarrow g(g(g(a))), 1.1\rangle$ is produced that
is $\Pi$-obsolete.  

The following lemmata state some properties of $\Pi$-stable rules.

\begin{lemma}
\label{lem_comp}
Let $\Pi$ be a set of forbidden patterns and let
$\langle l' = C[l \sigma]_p \rightarrow C[r \sigma]_p = r', p\rangle$ be a
$\Pi$-stable rewrite rule corresponding to $l \rightarrow r$. 
If $s \rightarrow t$ with $l' \rightarrow r'$, then $s \rightarrow_{\Pi}
t$ with $l \rightarrow r$.
\end{lemma}

\begin{lemma}
\label{lem_succ}
Let $\langle l \rightarrow r, p\rangle$ be a rule and $\Pi$
be a set of forbidden patterns. If 
$T(\langle l \rightarrow r, p\rangle) = \emptyset$, then
$\langle l \rightarrow r, p\rangle$ is either $\Pi$-stable
or $\Pi$-obsolete.
\end{lemma}

\begin{definition}
\label{trafo}
Let $\mathcal{R} = (\mathcal{F}, R)$ be a TRS with an associated set of forbidden patterns $\Pi$
where $\mathcal{F} = (S, \Omega)$.
The transformation $T$ maps TRSs with forbidden patterns to
standard TRSs $T(\mathcal{R}, \Pi)$. It proceeds in $5$ steps.
\begin{enumerate}
\item $R^{tmp} = \{\langle l \rightarrow r, \epsilon\rangle \mid l \rightarrow r \in R\}$\\
$R^{acc} = \emptyset$
\item \label{rec} $R^{acc} = \{\langle l \rightarrow r, p\rangle \in R^{tmp} \mid 
stb_{\Pi}(\langle l \rightarrow r, p\rangle) \}$ \\
$R^{tmp} = \{\langle l \rightarrow r, p\rangle \in R^{tmp} \mid 
\neg stb_{\Pi}(\langle l \rightarrow r, p\rangle) \wedge
\neg obs_{\Pi}(\langle l \rightarrow r, p\rangle) \}$
\item $R^{tmp} = \bigcup_{\langle l \rightarrow r, p \rangle \in R^{tmp}} 
T(\langle l \rightarrow r, p \rangle)$
\item If $R^{tmp} \not= \emptyset$ go to \ref{rec}
\item $T(\mathcal{R},\Pi) = (\mathcal{F} \uplus \{top_s \mid s \in S\}, \{l \rightarrow r \mid \langle l \rightarrow r, p\rangle \in R^{acc} \})$ 
\end{enumerate} 
\end{definition}

In the transformation rewrite rules are iteratively created and collected
in $R^{tmp}$ (temporary rules). Those rules that are $\Pi$-stable
and will thus be present in the final transformed system are collected
in $R^{acc}$ (accepted rules).

\begin{lemma}
\label{lemma_sound}
Let $\mathcal{R}$ be a rewrite system and $\Pi$ be a set of forbidden
 (linear $h$-)patterns. If $s \rightarrow_{\mathcal{R},\Pi} t$ for
 ground terms 
$s$ and $t$, then
$\topp(s) \rightarrow \topp(s)$ in $T(\mathcal{R}, \Pi)$.
\end{lemma}

\begin{theorem}
\label{ground-termination}
Let $\mathcal{R}$ be a TRS and $\Pi$ be a set of linear $here$-patterns.
We have $s \rightarrow_{\Pi}^+ t$ for ground terms $s$ and $t$ if and
only if $\topp(s) \rightarrow_{T(\mathcal{R}, \Pi)}^+ \topp(t)$. 
\end{theorem}
\begin{proof}
The result is a direct consequence of Lemmata \ref{lem_comp} and \ref{lemma_sound}.
\end{proof}

\begin{corollary}
\label{thm_soundness}
Let $\mathcal{R}$ be a TRS and $\Pi$ be a set of linear $h$-patterns.
$\mathcal{R}$ is ground terminating under $\Pi$ 
if and only if $T(\mathcal{R}, \Pi)$ is
ground terminating.
\end{corollary}

Note that the restriction to ground terms is crucial in Corollary \ref{thm_soundness}. Moreover, 
ground termination and general termination do not coincide in general 
for rewrite systems with forbidden patterns (observe that the same is
true for other important rewrite restrictions and strategies such as
the outermost strategy). 

\begin{example}
Consider the TRS $\mathcal{R} = (\mathcal{F}, R)$ given by $\mathcal{F} = \{a, f\}$
(where $a$ is a constant) and $R$ consisting of the rule
\begin{eqnarray*}
f(x) & \rightarrow & f(x).
\end{eqnarray*}
Moreover, consider the set of forbidden patterns $\Pi = \{
\langle f(a), \epsilon, h \rangle, \langle f(f(x)), \epsilon, h \rangle\}$.
Then $\mathcal{R}$ is not $\Pi$-terminating because we have 
$f(x) \rightarrow_{\Pi} f(x)$ but it is $\Pi$-terminating on
all ground terms, as can be shown by Theorem \ref{ground-termination}, 
since $T(\mathcal{R}, \Pi) = \emptyset$. 
\end{example}

\begin{example}
\label{ex2nd4}
Consider the TRS of Example \ref{ex2nd2}. We use two sorts $NatList$ and
$Nat$, with function symbol types
$\snd: NatList \rightarrow Nat$, 
$\inff: Nat \rightarrow NatList$, 
$\topp: NatList \rightarrow NatList$
(note that another ``$\topp$'' symbol of type $Nat \rightarrow Nat$ is
not needed here), 
$s: Nat \rightarrow Nat$, 
$0: Nat$, 
$\nil: NatList$ and $:$ of type $Nat, NatList \rightarrow NatList$. 
According to Definition \ref{trafo}, the rules of $T(\mathcal{R},
\Pi$) are:

\begin{equation}
 \nonumber\begin{tabular}[b]{r@{ $\rightarrow$ }l@{\;\;\;\;\;\;}r@{ $\rightarrow$ }l}
    $\snd(\inff(x))$ & $\snd(x:\inff(\s(x)))$ & $\snd(x:(y:zs))$ & $y$ \\
	 $\topp(\inff(x))$ & $\topp(x:\inff(\s(x)))$ & $\snd(x':\inff(x))$ & $\snd(x':(x:\inff(\s(x))))$ \\
	 $\topp(x' : \inff(x))$ & $\topp(x' : (x : \inff(\s(x))))$.
\end{tabular}
\end{equation}
\normalsize
This system is terminating (and termination can be verified
automatically, e.g.\ by AProVE \cite{aprove}). Hence, by Corollary
\ref{thm_soundness} also the TRS with forbidden patterns from Example 
\ref{ex2nd2} is ground terminating.
\end{example}

\begin{example}
\label{ex_app_term}
The TRS $\mathcal{R}$ and forbidden patterns $\Pi$ from Example \ref{ex_app}
yield the following system $T(\mathcal{R}, \Pi)$. For the sake of saving space
we abbreviate $\app$ by $\aaa$, $\take$ by $\ttt$ and $\inff$ by $\iii$.
\begin{equation}
 \nonumber\begin{tabular}[b]{r@{ $\rightarrow$ }l@{\;\;\;\;\;\;}r@{ $\rightarrow$ }l}
	$\topp(\iii(x))$ & $\topp(x : \iii(\s(x)))$ &
	$\ttt(y, \iii(x))$ & $\ttt(y, x : \iii(\s(x)))$ \\
	$\aaa(y, \iii(x))$ & $\aaa(y, x : \iii(\s(x)))$ &
	$\topp(\aaa(\iii(x), y))$ & $\topp(\aaa(x : \iii(\s(x)), y))$ \\
	$\ttt(\aaa(\iii(x), y), z)$ & $\ttt(\aaa(x : \iii(\s(x)), y), z)$ &
	$\ttt(z, \aaa(\iii(x), y))$ & $\ttt(z, \aaa(x : \iii(\s(x)), y))$ \\
	$\aaa(\aaa(\iii(x), y), z)$ & $\aaa(\aaa(x : \iii(\s(x)), y), z)$ &
	$\aaa(z, \aaa(\iii(x), y))$ & $\aaa(z, \aaa(x : \iii(\s(x)), y))$ \\
	$\topp(\aaa(x : xs, ys))$ & $\topp(x : \aaa(xs, ys))$ &
	$\ttt(z, \aaa(x : xs, ys))$ & $\ttt(z, x : \aaa(xs, ys))$ \\
	$\aaa(\aaa(x : xs, ys), z)$ & $\aaa(x : \aaa(xs, ys), z)$ &
	$\aaa(z, \aaa(x : xs, ys))$ & $\aaa(z, x : \aaa(xs, ys))$ \\
	$\aaa(x : \iii(zs), ys)$ & $x : \aaa(\iii(zs), ys)$ &
	$\aaa(x : \s(zs), ys)$ & $x : \aaa(\s(zs), ys)$ \\
	$\aaa(x : (y : zs), ys)$ & $x : \aaa(y : zs, ys)$ &
	$\aaa(\nil, x)$ & $x$ \\
	$\ttt(\s(x), y : ys)$ & $\ttt(x, ys)$ &
	$\ttt(0, y : ys)$ & $y$ \\
	$\ttt(x, \nil)$ & $0$ 
\end{tabular}
\end{equation}

This system is terminating (and termination can be verified
automatically, e.g.\ by AProVE \cite{aprove}). Hence, again by
Corollary 
\ref{thm_soundness} also the TRS with forbidden patterns from Example
\ref{ex_app} is ground terminating. 
\end{example}

%***************************************************************************
\section{Conclusion and Related Work}
\label{sec:conclusion-and-related-work}

We have presented and discussed a novel approach to rewriting with
context restrictions using 
forbidden patterns to specify forbidden/allowed positions in a term
rather than arguments of functions as it was done previously in
context-sensitivity. Thanks to  
their flexibility and parametrizability, forbidden patterns are applicable to
a wider class of TRSs than traditional methods. In particular,
position-based strategies and context-sensitive rewriting occur as
special cases 
of such patterns.  

For the TRSs in Examples \ref{ex2nd} and \ref{ex_app} nice operational
behaviours can be achieved by using rewriting with forbidden patterns.
The restricted reduction relation induced by the forbidden patterns
is terminating while still being powerful enough to compute (head-)
normal forms. 
When using simpler approaches such as position-based strategies or
context-sensitive rewriting in these examples, such operational
properties cannot be achieved.
For instance, consider Example \ref{ex2nd}. There is an infinite
reduction sequence starting from $\inff(x)$ with the property that
every term has exactly one redex. Thus, non-termination is preserved
under any reduction strategy (as strategies do not introduce new normal
forms by definition). On the other hand, in order to avoid this
infinite sequence using context-sensitive rewriting, we must set
$2 \not\in \mu(:)$ (regardless of any additional reduction strategy). 
But in this case $\rightarrow_{\mu}$ does
not compute head-normal forms.

In \cite{ppdp01-lucas} \emph{on-demand rewriting} was introduced, which
is able to properly deal with the TRS of Example \ref{ex2nd}. This
means that with the on-demand rewriting the reduction relation induced
by the TRS of Example \ref{ex2nd} can be restricted in a way such that
it becomes 
terminating while still normal forms w.r.t.\ the restricted relation are
head-normal forms w.r.t.\ the unrestricted one. Indeed, Example
\ref{ex2nd} 
was the main motivating example for the introduction of on-demand rewriting
in \cite{ppdp01-lucas}.

However, for Example \ref{ex_app} we get that by restricting rewriting
by the proposed forbidden patterns we obtain a terminating relation that
is able to compute the normal forms of all well-formed ground terms. As the system
is 
orthogonal, any outermost-fair reduction strategy, e.g.\ parallel
outermost, is normalizing. Yet, by using such a strategy the
relation still 
remains non-terminating. 
In particular, our forbidden patterns
approach yields an effective procedure for deciding whether
a ground term is normalizing or not (it is not normalizing if its
$\rightarrow_{\Pi}$-normal form is not an $\rightarrow$-normal
form) for this example. 

On the other hand, by using context-sensitive
rewriting, termination can only be obtained if 
$2 \not\in \mu(:)$ which in
turn implies that 
the term $0 : \app(\nil, \nil)$ cannot be normalized despite having
a normal form $0 : \nil$.

For Examples \ref{ex2nd} and \ref{ex_app} effective strategies like
parallel outermost or $\mathcal{S}_{\omega}$ of \cite{middeldorp}
are normalizing (though under either strategy there are still
infinite derivations). We provide another example for which
these strategies fail to provide normalization while
the use of appropriate forbidden patterns yields normalization
(and termination) 

\begin{example}
\label{parout}
Consider the TRS $\mathcal{R}$ consisting of the following rules
\begin{equation}
 \nonumber\begin{tabular}[b]{r@{ $\rightarrow$ }l@{\;\;\;\;\;\;}r@{ $\rightarrow$ }l@{\;\;\;\;\;\;}r@{ $\rightarrow$ }l}
    $a$ & $b$ & $b$ & $a$ & $c$ & $c$ \\
    $g(x, x)$ & $d$ & $f(b, x)$ & $d$
\end{tabular}
\end{equation}
Using a parallel outermost strategy the term $g(a, b)$ is not
reduced to its (unique) normal form $d$. Using $\mathcal{S}_{\omega}$,
$f(a, c)$ is not reduced to its (unique) normal form $d$.

However, it is easy to see that when using a $\Pi = \{\langle c, \epsilon, h\rangle, 
\langle b, \epsilon, h\rangle\}$, $\rightarrow_{\Pi}$ is terminating
and all $\mathcal{R}$-normal forms can be computed.
\end{example}

Note however, that the forbidden patterns used in Example \ref{parout}
are not canonical. Thus it is not clear how to come up with
such patterns automatically. 

We argued that for our forbidden pattern approach it is crucial to identify reasonable
classes of patterns that provide trade-offs between practical feasibility,
simplicity and power, favoring either component to a certain degree.
We have sketched and illustrated two approaches to deal with
the issues of verifying termination and guaranteeing that it is possible to
compute useful results (in our case original head-normal forms) with
the restricted rewrite 
relation.
To this end we proposed a transformation from rewrite systems with forbidden patterns
to ordinary rewrite systems and showed that ground termination of both induced
reduction relations coincide.
Moreover, we provided a criterion based on canonical rewriting with forbidden
patterns to ensure that normal forms w.r.t.\ the restricted reduction relation
are original head-normal forms.

In particular ``here''-patterns seem interesting as their use avoids context restrictions
to be \emph{non-local}. That is to say that whether a position is
allowed for reduction or not depends only on a restricted ``area'' around
the position in question regardless of the actual size of the whole
object term. Note that this is not true for ordinary context-sensitive
rewriting 
and has led to various complications in the theoretical analysis
(cf.~e.g.~
\cite[Definition 23]{jfp04-giesl-middeldorp}
\cite[Definition 7]{lpar08-alarcon-et-al} and
\cite[Definitions 1-3]{rta06-gramlich-lucas}).

Regarding future work, among many interesting questions and problems
one particularly important aspect is to identify conditions and
methods for the automatic (or at least automatically supported)
synthesis of appropriate forbidden pattern restrictions. 

~\\[+1ex]
\textbf{Acknowledgements}: We are grateful to the anonymous referees 
for numerous helpful and detailed comments and criticisms.

\bibliographystyle{abbrv}

\end{document}